\documentclass[11pt]{article}

\usepackage[pagebackref,colorlinks,hypertexnames=false]{hyperref}
\usepackage{orcidlink}

\usepackage{amsmath} 
\usepackage{amsthm} 
\usepackage{amssymb}	
\usepackage{graphicx} 
\usepackage{multicol} 
\usepackage{multirow}
\usepackage{color}
\usepackage[dvips,letterpaper,margin=1in,bottom=1in]{geometry}
\usepackage[capitalize,noabbrev]{cleveref}

\usepackage[utf8]{inputenc}
\usepackage[english]{babel}

\usepackage{bm}

\usepackage{mathtools}
\makeatletter
\newcommand{\vast}{\bBigg@{3}}
\makeatother

\usepackage{thmtools}

\usepackage{adjustbox}

\newtheorem{theorem}{Theorem}[section]

\newtheorem{lemma}[theorem]{Lemma}
\newtheorem{corollary}[theorem]{Corollary}
\newtheorem{proposition}[theorem]{Proposition}

\newtheorem{definition}[theorem]{Definition}

\newcommand{\ketbra}[3][]{\left\lvert #2 \vphantom{#3} \right>_{#1}\!\!\left< #3 \vphantom{#2} \right\rvert}

\DeclarePairedDelimiter\rbra{\lparen}{\rparen}
\DeclarePairedDelimiter\sbra{\lbrack}{\rbrack}
\DeclarePairedDelimiter\cbra{\{}{\}}
\DeclarePairedDelimiter\abs{\lvert}{\rvert}
\DeclarePairedDelimiter\Abs{\lVert}{\rVert}

\DeclarePairedDelimiter\ket{\lvert}{\rangle}
\DeclarePairedDelimiter\bra{\langle}{\rvert}
\DeclarePairedDelimiter\ave{\langle}{\rangle}

\newcommand{\tr} {\operatorname{tr}}

\newcommand{\diag} {\operatorname{\mathsf{diag}}}
\newcommand{\polylog} {\operatorname{polylog}}

\newcommand{\Vol} {\operatorname{Vol}}
\newcommand{\Herm} {\operatorname{Herm}}

\newcommand{\norm}[2][]{\Abs*{#2}_{#1}}
\newcommand{\fnorm}[1]{\Abs*{#1}_F}

\newcommand{\Exps}[2][]{\mathop{\EE}_{#1}\sbra*{#2}}

\newcommand{\CC}{\mathbb{C}}
\newcommand{\EE}{\mathbb{E}}
\newcommand{\NN}{\mathbb{N}}
\newcommand{\RR}{\mathbb{R}}

\begin{document}

\title{Unitary Synthesis with Near-Optimal T-Count \\ for Near-Clifford Unitaries}

\author{
    Wang Fang\footnote{Wang Fang is with the School of Informatics, University of Edinburgh, EH8 9AB Edinburgh, United Kingdom (e-mail: \href{mailto:Wang.Fang@ed.ac.uk}{\nolinkurl{Wang.Fang@ed.ac.uk}}).}
    \and
    Chris Heunen\footnote{Chris Heunen is with the School of Informatics, University of Edinburgh, EH8 9AB Edinburgh, United Kingdom (e-mail: \href{mailto:Chris.Heunen@ed.ac.uk}{\nolinkurl{Chris.Heunen@ed.ac.uk}}).}
    \and
    Qisheng Wang\footnote{Qisheng Wang is with the School of Computer Science, Shanghai Jiao Tong University, Shanghai 200240, China (e-mail: \href{mailto:QishengWang1994@gmail.com}{\nolinkurl{QishengWang1994@gmail.com}}).}
}
\date{}

\maketitle

\begin{abstract}
    We present an approach to unitary synthesis that implements an arbitrary $n$-qubit unitary operator $U$ by a Clifford+T circuit with T-count $\widetilde{O}\rbra{2^n d_F^{\mathcal{C}}(U)}$,\footnote{In this paper, $\widetilde{O}\rbra{f} = O\rbra{f \polylog\rbra{f}}$.} where $d_F^{\mathcal{C}}(U)$ is the Frobenius norm distance of $U$ to the Clifford group. 
    The T-count is shown to be \textit{near-optimal} when $d_F^{\mathcal{C}}(U)$ is a constant. 
    Our approach improves the previous best upper bound $\widetilde{O}(2^{4n/3})$ due to \hyperlink{cite.Tan25}{Tan (2025)} for a large class of unitary operators $U$ as long as $d_F^{\mathcal{C}}(U) \ll 2^{n/3}$. 
\end{abstract}

\textbf{Keywords: Circuit synthesis, quantum computing, quantum circuits, unitary synthesis, T-count.}

\newpage

\tableofcontents
\newpage

\section{Introduction}
Unitary synthesis, also known as the synthesis of quantum circuits, is a fundamental problem in quantum computing. 
Early research \cite{BBC+95,Kni95,Cyb01,AS03,VMS04,MVBS04,SMB04,SBM06} focused on the gate count of exact synthesis, i.e., how many elementary quantum gates are needed to exactly implement an $n$-qubit unitary operator, and showed that $\Theta\rbra{2^{2n}}$ two-qubit gates are sufficient and necessary. 
For a universal discrete gate set, unitary synthesis to precision $\varepsilon$ can be done using $O\rbra{2^{2n}\polylog\rbra{1/\varepsilon}}$ two-qubit gates by the Solovay-Kitaev theorem \cite{DN06}.
Beyond total gate count, a growing body of work over the past several years has refined the unitary synthesis problem by considering circuit depth, ancilla qubits, and space-depth tradeoffs for CNOT circuits~\cite{JST+20}, state preparation~\cite{WZY24,STY+23,ZLY22,YZ23}, and general unitary synthesis~\cite{STY+23,YZ23}.
Related query-based formulations have also led to new upper and lower bounds for implementing arbitrary unitaries with query access to an appropriate oracle~\cite{Ros21,LMW23}.

Unitary synthesis by Clifford+T circuits is becoming more and more important, as a large number of quantum error correcting codes can be used to realize fault-tolerant quantum computation with the Clifford+T gate set. 
However, implementing T gates usually requires magic state distillation \cite{BK05,CTV17} and thus turns out to be much more costly than implementing Clifford gates. 
This difficulty is essentially inevitable due to the Eastin–Knill theorem \cite{EK09}. 
Despite recent progress in optimizing the cost of implementing T gates \cite{GSJ24}, it still dominates the total cost of implementing Clifford+T circuits.
On the other hand, an efficient classical simulation of Clifford circuits is known by the Gottesman-Knill theorem \cite{Got98} and was later extended to simulating Clifford+T circuits with few T gates \cite{BG16}. 
Therefore, T-count plays an important role in both fault-tolerant quantum computation and classical simulation of quantum circuits. 
This naturally raises the following question:
\begin{quotation}
    \emph{How many T gates are needed to implement an arbitrary unitary operator?}
\end{quotation}

Extensive study has focused on the case of one-qubit unitary operators \cite{KMM13,Sel15,ross2016optimal,beverland2020lower} and showed that $\Theta\rbra{\log\rbra{1/\varepsilon}}$ T gates are sufficient \cite{ross2016optimal} and necessary \cite{beverland2020lower}. 
Another special case is quantum state preparation, which specifies the first column of the unitary operator. 
An approach with T-count $O\rbra{\sqrt{2^n n \log\rbra{n/\varepsilon}} + \log^2\rbra{n/\varepsilon}}$ for quantum state preparation was proposed in \cite{Low2024tradingtgatesdirty} and was later improved to $\Theta\rbra{\sqrt{2^n \log\rbra{1/\varepsilon}}+\log\rbra{1/\varepsilon}}$ in \cite{gosset2024quantumstatepreparationoptimal}.
The result in \cite{gosset2024quantumstatepreparationoptimal} is obtained by an optimal diagonal unitary synthesis with the same T-count and further implies a T-count of $O\rbra{2^{3n/2}\sqrt{n\log\rbra{1/\varepsilon}} + 2^n n \log\rbra{1/\varepsilon}}$ for unitary synthesis by Clifford+T circuits. 
In addition, another approach to unitary synthesis with T-count $\widetilde{O}\rbra{2^{3n/2}}$ is implied in \cite[Theorem 1.5]{Ros21}.

A recent breakthrough \cite{Tan25} broke the $2^{3n/2}$-barrier of the T-count with an approach to unitary synthesis by Clifford+T circuits with T-count $O\rbra{2^{4n/3} L^{2/3} + 2^nL}$, where $L = n + \log\rbra{1/\varepsilon}$. 
However, only a lower bound of $\Omega\rbra{2^{n}\sqrt{\log\rbra{1/\varepsilon}}+\log\rbra{1/\varepsilon}}$ on the T-count was previously known in \cite[Theorem 4.3]{gosset2024quantumstatepreparationoptimal}. 

In this paper, we present a new approach for unitary synthesis. 
Specifically, we can implement an arbitrary $n$-qubit unitary operator $U$ by a Clifford+T circuit with T-count $\widetilde{O}\rbra{2^n d_F^{\mathcal{C}}(U) \polylog\rbra{1/\varepsilon}}$, where $d_F^{\mathcal{C}}(U)$ is the Frobenius norm distance of $U$ to the Clifford group $\mathcal{C}$ generated by the Hadamard gate $H$, the phase gate $S$, the CNOT gate. 
Moreover, 
\begin{enumerate}
    \item The T-count of our unitary synthesis is \textit{near-optimal} for a class of unitary operators $U$ that are constantly close to the Clifford group, i.e., $d_F^{\mathcal{C}}(U) = \Theta\rbra{1}$. 
    In comparison with the previous results, optimal approaches are only known for one-qubit \cite{ross2016optimal,beverland2020lower} and diagonal \cite{gosset2024quantumstatepreparationoptimal} unitary operators. 
    \item The T-count of our unitary synthesis improves the prior best $\widetilde{O}\rbra{2^{4n/3}\polylog\rbra{1/\varepsilon}}$ due to \cite{Tan25} for any unitary operator $U$ satisfying $d_F^{\mathcal{C}}(U) \ll 2^{n/3}$. 
\end{enumerate}
Technically, our approach suggests a new idea to unitary synthesis. 
In contrast to the previous approaches \cite{Low2024tradingtgatesdirty,gosset2024quantumstatepreparationoptimal,Tan25} that focus on the decomposition of unitary operators, we focus on the decomposition of Hamiltonians.
Specifically, we present a \textit{Low T-Count Decomposition} of any Hamiltonian $\mathsf{H}$, which allows us to implement the unitary $U = e^{-i\mathsf{H}}$ with low T-count. 
With this good decomposition of Hamiltonians, we can therefore synthesize any unitary operators by adopting quantum algorithms for Hamiltonian simulation (e.g., \cite{Low2019hamiltonian,GSLW19}) equipped with the Linear-Combination-of-Unitaries (LCU) algorithm \cite{childs2012hamiltonian}. 

\subsection{Main results}

We focus on the unitary synthesis by Clifford+T circuits. 
Here, a Clifford+T circuit is described by a sequence of the Hadamard gate $H$, the phase gate $S$, the CNOT gate, and the $\pi/8$ gate $T$ (also called the T gate). 
For the sake of rigor, we first clarify the notion of approximately implementing a unitary operator with another unitary operator that acts on a larger Hilbert space. 

\begin{definition}[Approximate implementation of unitaries]
    An $n$-qubit unitary operator $U$ is said to be implemented by an $\rbra{n+m}$-qubit unitary operator $V$ to precision $\varepsilon$ (in the diamond norm distance), if $\Abs{ \mathcal{U} - \mathcal{V} }_\diamond \leq \varepsilon$, where $\mathcal{U} \colon \rho \mapsto U\rho U^\dag$ and
    \[
        \mathcal{V} \colon \rho \mapsto \tr_{\textup{env}}\rbra[\Big]{ V \rbra[\big]{\rho \otimes \underbrace{\ketbra{0}{0}^{\otimes m}}_{\textup{env}}} V^\dag }.
    \]
\end{definition}

Our main theorem is a unitary synthesis with near-optimal T-count when the unitary operator is constantly close to the Clifford group. 

\begin{theorem} [Unitary synthesis with near-optimal T-count for near-Clifford unitaries] \label{corollary:near-Clifford-intro}
    Any $n$-qubit unitary $U$ can be implemented to precision $\varepsilon$ by a Clifford+T circuit using 
    \[
        \widetilde{O}\rbra*{2^n \max\cbra{d_F^{\mathcal{C}}(U), 1} \polylog\rbra{1/\varepsilon}}
    \]
    T gates and ancilla qubits, where 
    \[
    d_F^{\mathcal{C}}(U) = \min_{C \in \mathcal{C}} {} \min_{\theta \in \mathbb{R}} {} \Abs{U - e^{i\theta}C}_F
    \]
    is the distance of $U$ to the Clifford group and $\Abs{A}_F = \sqrt{\tr\rbra{A^\dag A}}$ is the Frobenius norm. 
    Furthermore, for $d_F^{\mathcal{C}}(U) \leq O\rbra{1}$, any such implementation of $U$ requires $\Omega\rbra{2^n \sqrt{\log\rbra{1/\varepsilon}}+\log\rbra{1/\varepsilon}}$ T gates (even with measurements and adaptivity). 
\end{theorem}

The unitary synthesis given in \cref{corollary:near-Clifford-intro} is optimal up to a polylogarithmic factor in $\varepsilon$ for all unitary $U$ with $d_F^{\mathcal{C}}(U) \leq O\rbra{1}$. 
The current best unitary synthesis due to \cite{Tan25} has T-count $\widetilde{O}\rbra{2^{4n/3}\polylog\rbra{1/\varepsilon}}$. 
Our unitary synthesis given in \cref{corollary:near-Clifford-intro} can achieve a better T-count when $d_F^{\mathcal{C}}(U) \ll 2^{n/3}$. 
Combining the result of \cite{Tan25}, we obtain a unitary synthesis with T-count
\[
    \widetilde{O}\rbra*{ 2^n \polylog\rbra{1/\varepsilon} \cdot \min\cbra*{2^{n/3}, d_F^{\mathcal{C}}(U) + 1} }. 
\]

\subsection{Techniques}

We summarize our technical contributions as follows. 
\begin{itemize}
    \item For the upper bound, our technical contribution is a Hamiltonian-based framework for unitary synthesis, which is achieved by a low T-count decomposition of Hamiltonians (see \cref{lemma:Hamiltonian-decomposition-intro}). 
    This decomposition is discovered by adopting the hypercontractivity of Boolean functions \cite{Bon68,Bon70} (see \cref{lemma:u-overlap}). 
    \item For the lower bound, we strengthen the result of \cite{gosset2024quantumstatepreparationoptimal} to the case where the unitary operator is constantly close to the identity operator. 
    This is achieved by adopting the lower bound for sphere packing \cite{Rog47,Sch58,Bal92,Ven13,CJMS23,schildkraut2024lower}. 
\end{itemize}

For illustration, we provide a simpler version of \cref{corollary:near-Clifford-intro} below. 

\begin{theorem}[Hamiltonian-based unitary synthesis with near-optimal T-count for near-identity unitaries] \label{thm:near-identity-intro}
    For any $n$-qubit Hamiltonian $\mathsf{H}$, we can implement the unitary operator $e^{-i\mathsf{H}}$ to precision $\varepsilon$ by a Clifford+T circuit using 
    \[
    \widetilde{O}\rbra*{2^n \max\cbra{\Abs{\mathsf{H}}_F, 1} \polylog\rbra{1/\varepsilon}}
    \]
    T gates and ancilla qubits.
    Furthermore, for $\Abs{\mathsf{H}}_F \leq O\rbra{1}$, any such implementation of $e^{-i\mathsf{H}}$ requires $\Omega\rbra{2^n \sqrt{\log\rbra{1/\varepsilon}}+\log\rbra{1/\varepsilon}}$ T gates (even with measurements and adaptivity). 
\end{theorem}

\cref{corollary:near-Clifford-intro} is an immediate corollary of the following \cref{thm:near-identity-intro}. 
The details of our approach are illustrated in the remainder of this section. 

\subsubsection{Upper bounds}

To implement the unitary operator $e^{-i\mathsf{H}}$ with the T-count upper bound in \cref{thm:near-identity-intro}, our preliminary idea is to approximately decompose the Hamiltonian $\mathsf{H}$ as a linear combination of unitary operators in the form
\begin{equation*}
    \mathsf{H} \approx \sum_{k=0}^{m-1} \alpha_k V_k. 
\end{equation*}
Then, the unitary operator $e^{-i\mathsf{H}}$ can be approximately implemented by Hamiltonian simulation (e.g., \cite{Low2019hamiltonian,GSLW19}) with LCU \cite{childs2012hamiltonian}. 
However, two major difficulties remain:
\begin{enumerate}
    \item The absolute sum of the coefficients $\alpha_k$, i.e., $\Abs{\alpha}_1$, should be small enough to apply the Hamiltonian simulation \cite{Low2019hamiltonian,GSLW19} at a low cost. 
    \item Each unitary term $V_k$ should be easy to implement so that their linear combinations can be implemented by LCU \cite{childs2012hamiltonian} at a low cost. 
\end{enumerate}

\paragraph{Low T-count decomposition.}
To this end, we provide a useful decomposition that addresses the above difficulties. 
\begin{lemma} [Low T-count decomposition, \cref{lem:arbitrary_operator_approximation} simplified] \label{lemma:Hamiltonian-decomposition-intro}
    For any $n$-qubit operator $\mathsf{H}$, we can approximately decompose $\mathsf{H}$ as a linear combination of $m = O\rbra{2^n \log\rbra{\Abs{\mathsf{H}}_F/\varepsilon}}$ unitary operators $V_0, V_1, \dots, V_{m-1}$ such that 
    \begin{equation*}
        \Abs*{\mathsf{H} - \widetilde{\mathsf{H}}}_F \leq \varepsilon, \text{ where } \widetilde{\mathsf{H}} = \sum_{k=0}^{m-1} \alpha_k V_k
    \text{ and }
        \Abs{\alpha}_1 \leq O\rbra{\Abs{\mathsf{H}}_F}.
    \end{equation*}
    Moreover, each $V_k$ is of the form \[B_1^{(k)}H^{\otimes n}D^{(k)}H^{\otimes n}B_2^{(k)},\] where $B_1^{(k)}, B_2^{(k)}, D^{(k)}$ are diagonal unitary operators. 
\end{lemma}

\cref{lemma:Hamiltonian-decomposition-intro} can be viewed as an extension of the Euler angle decomposition (cf.\ \cite[Theorem 4.1]{NC10}, see also \cite[Fact 2.4]{gosset2024quantumstatepreparationoptimal}) for high-dimensional non-unitary operators. 
Any one-qubit unitary $U$ can be composed as $U = AHBHC$, where $A, B, C$ are diagonal unitary operators and $H$ is the Hadamard gate. 
In \cref{lemma:Hamiltonian-decomposition-intro}, we show that any $n$-qubit operator $\mathsf{H}$ can be approximately decomposed as a linear combination of unitary operators of the form $A H^{\otimes n} B H^{\otimes n} C$, where $A, B, C$ are diagonal unitary operators.

Technically, \cref{lemma:Hamiltonian-decomposition-intro} is achieved by recursively applying the following lemma that finds a unitary operator of the specific form $A H^{\otimes n} B H^{\otimes n} C$ with large overlap with a given operator $\mathsf{H}$. 
This is obtained by adopting the hypercontractivity of Boolean functions.

\begin{lemma}[Unitary component with large overlap, \cref{lem:tr} restated] \label{lemma:u-overlap}
    For any $n$-qubit operator $\mathsf{H}$, there exist three diagonal unitary operators $B_1, B_2, D$ such that
    \begin{equation} \label{eq:HS-intro}
    \ave*{B_1H^{\otimes n}DH^{\otimes n}B_2, \mathsf{H}}_{\mathrm{HS}} \geq \frac{1}{18} \Abs{\mathsf{H}}_F,
    \end{equation}
    where $\ave{A, B}_{\mathrm{HS}} = \tr\rbra{A^\dag B}$ is the Hilbert-Schmidt inner product.
\end{lemma}
\begin{proof}[Proof sketch]
Let $B_1 = \diag\rbra{x}$ and $B_2 = \diag\rbra{y}$ and 
choose $D_{i,i}$ for each $i$ such that $D_{i,i} \cdot \bra{i} H^{\otimes n} B_1 \mathsf{H} B_2 H^{\otimes n} \ket{i} = \abs{\bra{i} H^{\otimes n} B_1 \mathsf{H} B_2 H^{\otimes n} \ket{i}}$.
Then,
\[
\ave*{B_1H^{\otimes n}DH^{\otimes n}B_2, \mathsf{H}}_{\mathrm{HS}} = \sum_{i=0}^{N-1} {\abs*{x^{\mathrm{T}}\mathsf{H}_i' y}}, 
\]
where $N = 2^n$ and $\rbra{\mathsf{H}_i'}_{j,k} = \rbra{H^{\otimes n}}_{i,j}\mathsf{H}_{j,k}\rbra{H^{\otimes n}}_{k,i}$.
On the other hand, the quadratic form $S_{A}\rbra{x, y} = x^{\mathrm{T}}A y$ is of degree $2$, and by Bonami's lemma \cite{Bon68,Bon70}, we have (see \cref{lem:1_norm_xy})
\begin{align*}
    \Exps[x,y\sim \cbra{\pm 1}^N] { \abs*{S_A\rbra{x, y}} } &\geq \frac{1}{18} \sqrt{ \Exps[x,y\sim \cbra{\pm 1}^N] { \abs*{S_A\rbra{x, y}}^2 } }
    = \frac{1}{18} \Abs{A}_F.
\end{align*}
Therefore, 
\begin{align*}
    \Exps[x,y\sim \cbra{\pm 1}^N] { \sum_{i=0}^{N-1} {\abs*{x^{\mathrm{T}}\mathsf{H}_i' y}} } &\geq \sum_{i=0}^{N-1} \frac{1}{18} \Abs{\mathsf{H}_i'}_F 
    = \sum_{i=0}^{N-1} \frac{1}{18} \cdot \frac{1}{N}\Abs{\mathsf{H}}_F 
    = \frac{1}{18} \Abs{\mathsf{H}}_F. 
\end{align*}
This means that there always exists a pair of $x$ and $y$ such that \cref{eq:HS-intro} holds. 
\end{proof}

\paragraph{Unitary synthesis by Hamiltonian simulation.}
\cref{lemma:Hamiltonian-decomposition-intro} enables us to implement $e^{-i\mathsf{H}}$ as follows. Without loss of generality, we assume that $m = 2^a$ for some $a \geq 1$. For simplicity, we omit the polylogarithmic factors in $\Abs{\mathsf{H}}_F$ and $\varepsilon$ here.
\begin{itemize}
    \item \textbf{Step 1.} Implement a unitary $U = \rbra{G^\dag \otimes I} V \rbra{G \otimes I}$ such that $\bra{0}^{\otimes a} U \ket{0}^{\otimes a} = \widetilde{\mathsf{H}}/\Abs{\alpha}_1$, where 
    \begin{equation*}
        G \ket{0}^{\otimes a} = \sum_{k=0}^{m-1} \sqrt{\frac{\alpha_k}{\Abs{\alpha}_1}} \ket{k}
        \text{ and }
        V = \sum_{k=0}^{m - 1} \ketbra{k}{k} \otimes V_k.
    \end{equation*}
    Here, $G$ can be viewed as an $a$-qubit state-preparation unitary and thus can be implemented with T-count $O\rbra{\sqrt{2^a}}$ by the quantum state preparation (\cref{lem:optimal_preparation}) of \cite{Low2024tradingtgatesdirty,gosset2024quantumstatepreparationoptimal}. 
    Moreover, $V$ has the form
    \begin{align*}
        V ={}& \rbra[\Bigg]{\underbrace{\bigoplus_{k=0}^{m-1} B_1^{(k)}}_{\coloneqq B_1}} \rbra*{I_{m}\otimes H^{\otimes n}} \rbra[\Bigg]{\underbrace{\bigoplus_{k=0}^{m-1} D^{(k)}}_{\coloneqq D}}
        \rbra*{I_{m}\otimes H^{\otimes n}} \rbra[\Bigg]{\underbrace{\bigoplus_{k=0}^{m-1} B_2^{(k)}}_{\coloneqq B_2}},
    \end{align*}
    where $B_1, B_2, D$ are $\rbra{n+a}$-qubit diagonal unitary operators and can be implemented with T-count $O\rbra{\sqrt{2^{n+a}}}$ by the diagonal unitary synthesis (\cref{lem:optimal_diagonal}) of \cite{gosset2024quantumstatepreparationoptimal}. 
    In summary, $U$ can be implemented with T-count $O\rbra{\sqrt{2^a} + \sqrt{2^{n+a}}} = O\rbra{\sqrt{2^{n+a}}}$.

    \item \textbf{Step 2.} Implement $e^{-i\mathsf{H}} \approx e^{-i \rbra{\widetilde{\mathsf{H}}/\Abs{\alpha}_1} t}$ by the Hamiltonian simulation (\cref{thm:hamiltonian_simulation}) of \cite{Low2019hamiltonian,GSLW19} with simulation time $t = \Abs{\alpha}_1$. 
    As the Hamiltonian simulation uses $O\rbra{t}$ queries to (controlled-)$U$ and (controlled-)$U^\dag$, the unitary $e^{-i\mathsf{H}}$ can thus be implemented with T-count $\widetilde{O}\rbra{t\sqrt{2^{n+a}}} = \widetilde{O}\rbra{2^n \Abs{\mathsf{H}}_F}$. 
\end{itemize}

With detailed analysis, the above approach can implement $e^{-i\mathsf{H}}$ to precision $\varepsilon$ with T-count $\widetilde{O}\rbra{2^n \Abs{\mathsf{H}}_F \log^2\rbra{1/\varepsilon}}$ for any Hamiltonian $\mathsf{H}$. 

\subsubsection{Lower bounds}

Our T-count lower bound for the synthesis of unitary $e^{-i\mathsf{H}}$ with $\Abs{\mathsf{H}}_F \leq O\rbra{1}$ builds on the fact (\cref{lem:state_count}) in \cite{gosset2024quantumstatepreparationoptimal} that any $n$-qubit Clifford circuit with Pauli postselections using $t$ copies of the magic state $\ket{T} = \frac{1}{\sqrt{2}}\rbra{\ket{0}+e^{i\pi/4}\ket{1}}$ can prepare at most $2^{O\rbra{n^2+t^2}}$ different $n$-qubit states. 
The T-count lower bound in \cite{gosset2024quantumstatepreparationoptimal} for unitary synthesis is obtained by providing $\rbra{1/\varepsilon}^{\Theta\rbra{4^n}}$ $n$-qubit unitary operators with trace distance between their Choi states at least $\varepsilon$, which gives $2^{O\rbra{n^2+t^2}} \geq \rbra{1/\varepsilon}^{\Theta\rbra{4^n}}$, that is, $t \geq \Omega\rbra{2^n \sqrt{\log\rbra{1/\varepsilon}}}$. 

Our approach strengthens the T-count lower bound in \cite{gosset2024quantumstatepreparationoptimal} by considering the unitary operators $U = e^{-i\mathsf{H}}$ with constraints $\Abs{\mathsf{H}}_F \leq O\rbra{1}$. 
Our strategy is to find a collection of unitary operators of the form $e^{-i\mathsf{H}}$ with $\Abs{\mathsf{H}}_F \leq 1$ such that the trace distance between their Choi states is at least $\varepsilon$. 
To this end, we adopt the lower bound for sphere packing given in \cite{schildkraut2024lower} and thus can find a finite set $A \subset \{\mathsf{H} \mid \mathsf{H} = \mathsf{H}^\dag \textup{ and } \Abs{\mathsf{H}}_F \leq 1\}$ such that $\abs{A} \geq \rbra{\sqrt{2^n}\varepsilon}^{-\Omega\rbra{4^n}}$ and for any two distinct $\mathsf{H}, \mathsf{H}' \in A$, the trace distance between the Choi states of $e^{-i\mathsf{H}}$ and $e^{-i\mathsf{H}'}$ is at least $\varepsilon$.
With the same argument as in \cite{gosset2024quantumstatepreparationoptimal}, we have $2^{O\rbra{n^2+t^2}} \geq \rbra{\sqrt{2^n}\varepsilon}^{-\Omega\rbra{4^n}}$, which gives $t \geq \Omega\rbra{2^n \sqrt{\log\rbra{1/\varepsilon}}}$ for sufficiently small $\varepsilon > 0$. 

Together with the lower bound $\Omega\rbra{\log\rbra{1/\varepsilon}}$ in \cite[Lemma~5.9]{beverland2020lower}, we therefore obtain a T-count lower bound of $\Omega\rbra{2^n \sqrt{\log\rbra{1/\varepsilon}} + \log\rbra{1/\varepsilon}}$ for the synthesis of all unitary operators $e^{-i\mathsf{H}}$ with $\Abs{\mathsf{H}}_F \leq O\rbra{1}$. 

\subsection{Discussion}

In this paper, we propose a new approach to unitary synthesis based on a low-T-count Hamiltonian decomposition, which is near-optimal for near-Clifford unitary operators.
We conclude this section by listing two open questions for future research.
Can our Hamiltonian-based approach be improved for unitary operators that are not near-Clifford? Are there alternative strategies that yield better results for general unitary synthesis?
Can the T-count lower bound for unitary synthesis be further improved?

\section{Preliminaries}\label{sec:preliminary}

This section introduces the notation used throughout the paper and presents the key tools for our analysis.

\subsection{Basic notation}
For any positive integer $n$, we write $[n]$ for the set $\{1,2,\ldots, n\}$.
For any matrix $A$, denote the operator norm by $\norm{A} = \sup\{ \|Ax\| \mid x \in \CC^n, \|x\|\leq 1\}$, the Frobenius norm by $\fnorm{A} = \sqrt{\tr\rbra{A^\dag A}}$, and the trace norm by $\norm[1]{A} = \tr\rbra{\sqrt{A^\dag A}}$.
For any superoperator $\mathcal{E}$, the diamond norm of $\mathcal{E}$ is defined as 
\[
\Abs{\mathcal{E}}_\diamond = \max_{\Abs{X}_1 \leq 1} \Abs*{ \rbra*{\mathcal{E} \otimes \mathcal{I}}\rbra{X} }_1,
\]
where $\mathcal{I}$ is the identity superoperator of the same dimension as $\mathcal{E}$.
The induced trace norm of $\mathcal{E}$ is defined as
\[
\Abs{\mathcal{E}}_1 = \max_{\Abs{X}_1 \leq 1} \Abs*{\mathcal{E}\rbra{X}}_1.
\]
For any function $f$ from $\{\pm 1\}^n$ to $\RR$ or $\CC$, the $q$-norm of $f$ is defined as
\begin{equation*}
    \norm[q]{f} \coloneq \rbra*{\Exps[x\sim \{\pm 1\}^n]{\abs*{f(x)}^q}}^{1/q},
\end{equation*}
where $x\sim \{\pm 1\}^n$ denotes that $x$ is drawn uniformly at random from $\{\pm 1\}^n$, and $\EE$ denotes the expectation.
For any vector $v$, $\diag(v)$ denotes the diagonal matrix with $v$ as its diagonal entries.

\subsection{Hypercontractivity of Boolean functions}
Real-valued Boolean functions satisfy the following hypercontractivity.

\begin{lemma}
    [Bonami's lemma \cite{Bon68,Bon70}]\label{lem:bonami}
    Let $f\colon\cbra{\pm 1}^n \to \RR$ be a polynomial of degree $d$. Then
    \begin{equation*}
        \norm[4]{f} \leq 3^{d/2} \norm[2]{f}.
    \end{equation*}
\end{lemma}

This hypercontractivity naturally extends to complex-valued Boolean functions, which our approach uses.

\begin{corollary}\label{cor:bonami_complex}
    Let $f\colon \cbra{\pm 1}^n \to \CC$ be a polynomial of degree $d$. Then
    \begin{equation*}
        \norm[4]{f} \leq \sqrt{2}\cdot 3^{d/2} \norm[2]{f}.
    \end{equation*}
\end{corollary}
\begin{proof}
    Write $f=g+ih$ for polynomials $g, h\colon \cbra{\pm 1}^n \to \RR$ of degree $d$.
    Applying \cref{lem:bonami} to $g$ and $h$ separately gives
    \begin{equation*}
        \norm[4]{g} \leq 3^{d/2}\norm[2]{g}, \qquad \norm[4]{h} \leq 3^{d/2}\norm[2]{h}. 
    \end{equation*}
    Now
    \begin{align*}
        \norm[4]{f} &\leq \norm[4]{g}+\norm[4]{h} 
        \leq 3^{d/2}\rbra*{\norm[2]{g}+\norm[2]{h}}
        \leq 3^{d/2}\sqrt{2\rbra*{\norm[2]{g}^2+\norm[2]{h}^2}}
        = \sqrt{2}\cdot 3^{d/2}\norm[2]{f}. \qedhere
    \end{align*}
\end{proof}

Then, a standard trick using H\"older's inequality gives a lower bound on $1$-norm.
\begin{lemma}\label{lem:1_norm_lower_bound}
    Let $f\colon \cbra{\pm 1}^n \to \CC$ be a polynomial of degree $d$. Then
    \begin{equation*}
        \norm[1]{f} \geq \frac{1}{2\cdot 3^d} \norm[2]{f}.
    \end{equation*}
\end{lemma}
\begin{proof}    
    Consider $\abs{f}^2 = \abs{f}^{4/3}\abs{f}^{2/3}$, by H\"older's inequality, we have
    $\norm[3]{\abs{f}^{\frac{4}{3}}}\norm[\frac{3}{2}]{\abs{f}^{\frac{2}{3}}} \geq \norm[1]{\abs{f}^2}$,
    that is,
    $\norm[4]{f}^{\frac{4}{3}}\norm[1]{f}^{\frac{2}{3}} \geq \norm[2]{f}^2$. 
    With \cref{cor:bonami_complex}, we have
    \begin{equation*}
        \norm[1]{f} \geq \frac{\norm[2]{f}^3}{\norm[4]{f}^2} \geq \frac{\norm[2]{f}^3}{\rbra*{\sqrt{2}\cdot 3^{d/2}\norm[2]{f}}^2} = \frac{1}{2\cdot 3^d}\norm[2]{f}. \qedhere
    \end{equation*}
\end{proof}

\subsection{Diagonal unitary synthesis and state preparation with optimal T-count}
As preparation for our near-Clifford unitary synthesis, we also review the existing results on T-count of unitary synthesis and state preparation, which serve as the foundation for our approach.

\begin{lemma}
    [Single-qubit unitary synthesis with optimal T-count~\cite{ross2016optimal}]\label{lem:optimal_single_qubit}
    Any single-qubit unitary with determinant $1$ can be implemented up to error $\varepsilon$ (in operator norm distance) by a Clifford+T circuit using $O(\log(1/\varepsilon))$ many T gates and without ancilla qubits.
\end{lemma}

\begin{corollary}
    [Boolean phase oracle synthesis~{\cite[Fact~3.3]{gosset2024quantumstatepreparationoptimal}}]\label{lem:optimal_boolean_phase}
    Let $B$ be an $n$-qubit Boolean phase oracle, that is, a diagonal unitary all of whose diagonal entries are $\pm 1$). Then $B$ can be implemented exactly by a Clifford+T circuit using $O\rbra*{\sqrt{2^n}}$ many T gates and ancillas.
\end{corollary}

\begin{lemma}
    [Diagonal unitary synthesis with optimal T-count~{\cite[Theorem~1.2]{gosset2024quantumstatepreparationoptimal}}]\label{lem:optimal_diagonal}
    Any diagonal unitary on $n$ qubits can be implemented up to error $\varepsilon$ (in operator norm distance) by a Clifford+T circuit using
    \begin{equation*}
        O\rbra*{\sqrt{2^n\log\rbra*{1/\varepsilon}}+\log\rbra*{1/\varepsilon}}
    \end{equation*}
    T gates and ancillas.
    Furthermore, no Clifford+T circuit (even with measurements and adaptivity) can use asymptotically fewer T gates.
\end{lemma}

\begin{lemma}
    [Quantum state preparation with optimal T-count~{\cite[Theorem~1.1]{gosset2024quantumstatepreparationoptimal}}]\label{lem:optimal_preparation}
    Any $n$-qubit state can be prepared up to error $\varepsilon$ (in 2-norm distance) by a Clifford+T circuit starting with the all-zeros state using
    \begin{equation*}
        O\rbra*{\sqrt{2^n\log\rbra*{1/\varepsilon}}+\log\rbra*{1/\varepsilon}}
    \end{equation*}
    T gates and ancillas.
    Furthermore, no Clifford+T circuit (even with measurements and adaptivity) can use asymptotically fewer T gates.
\end{lemma}

\subsection{Diamond norm bound for block-encoded unitaries}
Finally, we provide a bound relating the diamond norm and the block-encoding's operator norm.

\begin{lemma}\label{lem:block_encoding_diamond}
    Let $V$ be an $n$-qubit unitary operator. 
    Suppose that $U$ is an $\rbra{n+a}$-qubit unitary operator such that $\Abs{\bra{0}^{\otimes a} U \ket{0}^{\otimes a} - V} \leq \varepsilon$. 
    Then, $\Abs{\mathcal{U} - \mathcal{V}}_\diamond \leq 2^{n+2}\varepsilon$, where $\mathcal{V} \colon \rho \mapsto V \rho V^\dag$ and
    \begin{equation*}
        \mathcal{U} \colon \rho \mapsto \tr_a\rbra*{ U \rbra*{ \rho \otimes \ketbra{0}{0}^{\otimes a} } U^\dag }.
    \end{equation*}
\end{lemma}
\begin{proof}
    \allowdisplaybreaks
    For any $n$-qubit density operator $\rho$, we have
    \begin{align*}
        \mathcal{U}\rbra{\rho}
        &= \tr_a\rbra*{ U \rbra*{ \rho \otimes \ketbra{0}{0}^{\otimes a} } U^\dag }
        = \sum_{j=0}^{2^a-1} \bra{j}_a U \ket{0}_a \cdot \rho \cdot \bra{0}_a U \ket{j}_a,
    \end{align*}
    so
    \begin{align*}
        \Abs{ \mathcal{U} - \mathcal{V} }_{1}
        ={}& \max_{\rho} {} \Abs*{ \mathcal{U}\rbra{\rho} - \mathcal{V}\rbra{\rho} }_{1} \\
        ={}& \max_{\rho} {} \Abs*{ \sum_{j=0}^{2^a-1} \bra{j}_a U \ket{0}_a \cdot \rho \cdot \bra{0}_a U^\dag \ket{j}_a - V \rho V^\dag }_1 \\
        \leq{} & \max_{\rho} {} \rbra*{ \Abs*{ \bra{0}_a U \ket{0}_a \cdot \rho \cdot \bra{0}_a U^\dag \ket{0}_a - V \rho V^\dag }_1  + \Abs*{ \sum_{j=1}^{2^a-1} \bra{j}_a U \ket{0}_a \cdot \rho \cdot \bra{0}_a U^\dag \ket{j}_a }_1 } \\
        ={}& \max_{\rho} {} \rbra*{ \Abs*{ \bra{0}_a U \ket{0}_a \cdot \rho \cdot \bra{0}_a U^\dag \ket{0}_a - V \rho V^\dag }_1 + \tr\rbra*{ \rbra*{ I - \bra{0}_a U^\dag \ket{0}_a \cdot \bra{0}_a U \ket{0}_a} \rho } } \\
        \leq{}& \Abs*{ \bra{0}_a U^\dag \ket{0}_a \cdot \bra{0}_a U \ket{0}_a - V^\dag V } + \Abs*{I - \bra{0}_a U^\dag \ket{0}_a \cdot \bra{0}_a U \ket{0}_a} \\
        ={}& 2 \Abs*{V^\dag V - \bra{0}_a U^\dag \ket{0}_a \cdot \bra{0}_a U \ket{0}_a} \tag{$V^\dag V = I$} \\
        \leq{}& 2 \big( \Abs*{V^\dag V - V^\dag \cdot \bra{0}_a U \ket{0}_a}  + \Abs*{V^\dag \cdot \bra{0}_a U \ket{0}_a - \bra{0}_a U^\dag \ket{0}_a \cdot \bra{0}_a U \ket{0}_a} \big) \\
        \leq{}& 2 \big( \Abs*{V^\dag} \Abs*{V - \bra{0}_a U \ket{0}_a} + \Abs*{V^\dag - \bra{0}_a U^\dag \ket{0}_a} \Abs*{\bra{0}_a U \ket{0}_a} \big) \\
        \leq{}& 2 \rbra*{ \Abs*{V - \bra{0}_a U \ket{0}_a} + \Abs*{V^\dag - \bra{0}_a U^\dag \ket{0}_a} } \\
        ={}& 4 \Abs*{V - \bra{0}_a U \ket{0}_a} \tag{$\Abs{A^\dag} = \Abs{A}$} \\
        \leq{}& 4\varepsilon,
    \end{align*}
    and hence
    \[
    \Abs{\mathcal{U} - \mathcal{V}}_\diamond \leq 2^n \Abs{\mathcal{U} - \mathcal{V}}_1 \leq 2^{n+2} \varepsilon. \qedhere
    \]
\end{proof}

\section{Near-Clifford unitary synthesis with near-optimal T-count}
In this section, we present a method for approximating arbitrary operators through linear combinations of interleaved Hadamard operators and diagonal unitaries, which extends the approximation method for states~\cite{gosset2024quantumstatepreparationoptimal} to operators.
Then, applying our approximation method to the unitary's Hamiltonian, together with the Linear-Combination-of-Unitaries algorithm (LCU)~\cite{childs2012hamiltonian}, constructs a standard encoding oracle of this Hamiltonian with good T-count by \cref{lem:optimal_diagonal}.
Finally, the robust block-Hamiltonian simulation via QSVT~\cite{GSLW19} with this encoded Hamiltonian oracle completes the unitary synthesis.

\subsection{Approximating arbitrary operators}
We first observe that the quadratic form $x^{\mathrm{T}} A y$ for any matrix $A \in \CC^{N\times N}$ attains a a certain minimal magnitude when $x$ and $y$ are chosen uniformly at random from $\cbra{\pm 1}^N$.

\begin{lemma}\label{lem:1_norm_xy}
    Let $N$ be a positive integer and $A \in \CC^{N\times N}$. Then
    \begin{equation*}
        \Exps[x,y\sim \cbra{\pm 1}^N]{\abs*{\sum_{j\in\sbra{N}}\sum_{k\in\sbra{N}}A_{j,k}\cdot x_{j}\cdot y_{k}}} \geq \frac{1}{18}\fnorm{A}.
    \end{equation*}
\end{lemma}
\begin{proof}
    Let $S(x,y) = \sum_{j\in\sbra{N}}\sum_{k\in\sbra{N}}A_{j,k}\cdot x_{j}\cdot y_{k}$. Then
    \begin{align*}
        \Exps[x,y\sim \cbra{\pm 1}^N]{\abs*{S(x,y)}^2}
        ={}&
        \Exps[x,y\sim \cbra{\pm 1}^N]{\sum_{j,j',k,k'\in\sbra{N}}A_{j,k}A_{j',k'}^*\cdot x_{j}x_{j'}\cdot y_{k}y_{k'}} \\
        ={}& \sum_{j,j',k,k'\in\sbra{N}}A_{j,k}A_{j',k'}^*\cdot \Exps[x,y\sim \cbra{\pm 1}^N]{x_{j}x_{j'}\cdot y_{k}y_{k'}} \\
        ={}& \sum_{j,j',k,k'\in\sbra{N}}A_{j,k}A_{j',k'}^*\cdot \delta_{j,j'}\cdot \delta_{k,k'} \\
        ={}& \sum_{j,k\in\sbra{N}}\abs{A_{j,k}}^2 \\
        ={}& \norm[F]{A}^2,
    \end{align*}
    because
    $\Exps[x,y\sim \cbra{\pm 1}^N]{x_{j}x_{j'}\cdot y_{k}y_{k'}} = \delta_{j,j'}\cdot \delta_{k,k'}$
    as $x$ and $y$ are uniformly distributed over $\{\pm 1\}^N$.
    Since $S$ is a polynomial of degree $2$, by \cref{lem:1_norm_lower_bound}, we have
    \begin{align*}
        \Exps[x,y\sim \cbra{\pm 1}^N]{\abs*{S(x,y)}} 
        &\geq \frac{1}{18}\sqrt{\Exps[x,y\sim \cbra{\pm 1}^N]{\abs*{S(x,y)}^2}}
        = \frac{1}{18}\fnorm{A}. \qedhere
    \end{align*} 
\end{proof}

Using Hadamard operators to ensure uniform magnitude across all indices, we establish the following lemma.

\begin{lemma}\label{lem:abs_tr}
    Let $A \in \CC^{2^n\times 2^n}$ and $F = H^{\otimes n}$. Then
    \begin{align*}
        \Exps[x,y\sim \cbra{\pm 1}^{\{0,1\}^n}]{\sum_{i}\abs*{\bra{i}F\diag(x)A\diag(y)F\ket{i}}}
        \geq{}& \frac{1}{18}\fnorm{A}.
    \end{align*}
\end{lemma}
\begin{proof}
    \allowdisplaybreaks
    \begin{align*}
        & \Exps[x,y\sim \cbra{\pm 1}^{\{0,1\}^n}]{\sum_{i}\abs*{\bra{i}F\diag(x)A\diag(y)F\ket{i}}} \\
        ={}& \Exps[x,y\sim \cbra{\pm 1}^{\{0,1\}^n}]{\sum_{i}\abs*{\sum_{j,k}F_{i,j}\cdot x_j\cdot A_{j,k}\cdot y_k\cdot F_{k,i}}} \\
        ={}& \sum_{i}\Exps[x,y\sim \cbra{\pm 1}^{\{0,1\}^n}]{\abs*{\sum_{j,k}F_{i,j}A_{j,k}F_{k,i}\cdot x_jy_k}} \\
        \geq{}& \sum_{i}\frac{1}{18}\sqrt{\sum_{j,k}\abs*{F_{i,j}A_{j,k}F_{k,i}}^2} \tag{by \cref{lem:1_norm_xy}} \\
        ={}& \sum_{i}\frac{1}{18\cdot 2^n}\sqrt{\sum_{j,k}\abs*{A_{j,k}}^2} \tag{because $\abs{F_{i,j}} = \abs{F_{k,i}} = \frac{1}{\sqrt{2^n}}$} \\
        ={}& \sum_{i}\frac{1}{18\cdot 2^n}\fnorm{A} \\
        ={}& \frac{1}{18}\fnorm{A}. \tag*{\qedhere}
    \end{align*}
\end{proof}

By multiplying appropriate complex phases (modulus-one scalars) in all summation terms of \cref{lem:abs_tr} to eliminate the absolute value symbols, we derive the following lemma.

\begin{lemma}\label{lem:tr}
    For any $A \in \CC^{2^n\times 2^n}$, there exist two $n$-qubit Boolean phase oracles $B_1$ and $B_2$, and an $n$-qubit diagonal unitary $D$ such that
    \begin{equation*}
        \tr\rbra*{(B_1H^{\otimes n}DH^{\otimes n}B_2)^{\dagger} \cdot A} \geq \frac{1}{18}\fnorm{A}.
    \end{equation*}
\end{lemma}
\begin{proof}
    By \cref{lem:abs_tr}, there exist Boolean phase oracles $B_1$ and $B_2$ such that
    \begin{equation*}
        \sum_{i}\abs*{\bra{i}H^{\otimes n}B_1AB_2H^{\otimes n}\ket{i}} \geq \frac{1}{18}\fnorm{A}.
    \end{equation*}
    Let $d \in \CC^{\cbra{0,1}^n}$ be the vector such that
    \begin{equation}\label{eq:d_def}
        \bra{i}H^{\otimes n}B_1AB_2H^{\otimes n}\ket{i} \cdot d_i = \abs*{\bra{i}H^{\otimes n}B_1AB_2H^{\otimes n}\ket{i}}
    \end{equation}
    for any $i \in \cbra{0,1}^n$. Then $D = \diag(d)^{\dagger}$ is a diagonal unitary as $\abs{d_i} = 1$, and
    \begin{align*}
        \tr\rbra*{\rbra*{B_1H^{\otimes n}DH^{\otimes n}B_2}^{\dagger}\cdot A} 
        ={}& \tr\rbra*{H^{\otimes n}B_1AB_2H^{\otimes n}D^{\dagger}} \\
        ={}& \sum_i \bra{i}H^{\otimes n}B_1AB_2H^{\otimes n}D^{\dagger}\ket{i} \\
        ={}& \sum_i \bra{i}H^{\otimes n}B_1AB_2H^{\otimes n}\ket{i} \cdot d_i \\
        ={}&  \sum_i \abs*{\bra{i}H^{\otimes n}B_1AB_2H^{\otimes n}\ket{i}} \tag{by \cref{eq:d_def}} \\
        \geq{}& \frac{1}{18}\fnorm{A}. \tag*{\qedhere}
    \end{align*}
\end{proof}

According to \cref{lem:tr}, for any operator, we can find efficiently implementable unitaries with substantial overlap in the Hilbert-Schmidt inner product.
Thus, using \cref{lem:tr} recursively, we can approximate any operator as in the following lemma.
The proof employs reasoning similar to~\cite{rosenthal2024efficient,gosset2024quantumstatepreparationoptimal}, but chooses adaptive step size at each iteration.

\begin{lemma}\label{lem:arbitrary_operator_approximation}
    For any $n$-qubit operator $A$ and any integer $m\geq 1$, let
    \[\beta = \sqrt{1 - \frac{1}{324\cdot 2^n}}.\]
    Then, there exist $n$-qubit unitaries $V_0, V_1, \ldots, V_{m-1}$ and non-negative numbers $\alpha_0, \alpha_1, \ldots, \alpha_{m-1}$ such that
    \begin{equation*}
        \fnorm{A - \sum_{k=0}^{m-1}\alpha_k V_k} \leq \beta^{m}\fnorm{A}
    \end{equation*}
    and
    \begin{equation*}
        \sum_{k=0}^{m-1} \alpha_k \leq \frac{1-\beta^m}{18\cdot 2^n(1-\beta)}\fnorm{A},
    \end{equation*}
    where each $V_k = B_1^{(k)}H^{\otimes n}D^{(k)}H^{\otimes n}B_2^{(k)}$ for some $n$-qubit Boolean phase oracles $B_1^{(k)}, B_2^{(k)}$, and some $n$-qubit diagonal unitaries $D^{(k)}$.
\end{lemma}
\begin{proof}
    We prove this lemma by induction on $m$.
    First consider the base case $m=1$.
        By \cref{lem:tr}, there exist $V_0 = B_1^{(0)}H^{\otimes n}D^{(0)}H^{\otimes n}B_2^{(0)}$ for two $n$-qubit Boolean phase oracles $B_1^{(0)}, B_2^{(0)}$, and an $n$-qubit diagonal unitary $D^{(0)}$ such that
        \begin{equation*}
            \tr\rbra*{V_0^{\dagger}\cdot A} \geq \frac{1}{18}\fnorm{A}.
        \end{equation*}
        Choosing
        \begin{equation}\label{eq:alpha_0}
            \alpha_0 = \frac{1}{18\cdot 2^n}\fnorm{A}
        \end{equation}
        gives
        \begin{align*}
            \fnorm{A - \alpha_0 V_0}^2
            ={}& \fnorm{A}^2 + \alpha_0^2\fnorm{V_0}^2 - \alpha_0\rbra*{\tr\rbra[\big]{V_0^{\dagger}A}+\tr\rbra[\big]{A^{\dagger}V_0}} \\
            ={}& \fnorm{A}^2 + \alpha_0^2\cdot 2^n - 2\alpha_0\tr\rbra*{V_0^{\dagger}A} \\
            \leq{}& \fnorm{A}^2 + \alpha_0^2\cdot 2^n - 2\alpha_0 \cdot \frac{1}{18}\fnorm{A} \\
            ={}& \fnorm{A}^2 + \frac{1}{324\cdot 2^n}\fnorm{A}^2 - \frac{2}{324\cdot 2^n}\fnorm{A}^2 \\
            ={}& \beta^2\fnorm{A}^2,
        \end{align*}
        which implies $\fnorm{A - \alpha_0 V_0} \leq \beta\fnorm{A}$. With \cref{eq:alpha_0}, we complete the basic case.
  
    Now for the inductive step. Suppose $m = j+1$ with $j \geq 1$. By the inductive hypothesis, there exist $n$-qubit unitaries $V_0, V_1, \ldots, V_{j-1}$ and positive numbers $\alpha_0, \alpha_1, \ldots, \alpha_{j-1}$ such that
        \begin{equation}\label{eq:hypothesis_A}
            \fnorm{A - \sum_{k=0}^{j-1}\alpha_k V_k} \leq \beta^{j}\fnorm{A}
        \end{equation}
        and
        \begin{equation}\label{eq:hypothesis_alpha}
            \sum_{k=0}^{j-1} \alpha_k \leq \frac{1-\beta^j}{18\cdot 2^n(1-\beta)}\fnorm{A},
        \end{equation}
        where each $V_k = B_1^{(k)}H^{\otimes n}D^{(k)}H^{\otimes n}B_2^{(k)}$ for some $n$-qubit Boolean phase oracles $B_1^{(k)}, B_2^{(k)}$, and some $n$-qubit diagonal unitaries $D^{(k)}$.

        Let $A_j = A - \sum_{k=0}^{j-1}\alpha_k V_k$. By \cref{lem:tr}, there exist $V_j = B_1^{(j)}H^{\otimes n}D^{(j)}H^{\otimes n}B_2^{(j)}$ for two $n$-qubit Boolean phase oracles $B_1^{(j)}, B_2^{(j)}$, and an $n$-qubit diagonal unitary $D^{(j)}$ such that
        \begin{equation*}
            \tr\rbra*{V_j^{\dagger}\cdot A_j} \geq \frac{1}{18}\fnorm{A_j}.
        \end{equation*}
        Choose
        \begin{equation}\label{eq:alpha_j}
            \alpha_j = \frac{1}{18\cdot 2^n}\fnorm{A_j} \leq \frac{\beta^j}{18\cdot 2^n}\fnorm{A}.
        \end{equation}
        Then,
        \begin{align*}
            \fnorm{A - \sum_{k=0}^{j}\alpha_k V_k}^2
            ={}& \fnorm{A_j - \alpha_jV_j}^2 \\
            ={}& \fnorm{A_j}^2 + \alpha_j^2\fnorm{V_j}^2 -2\alpha_j\tr\rbra*{V_j^{\dagger}A_j}\\
            \leq{}& \fnorm{A_j}^2 + \alpha_j^2\cdot 2^n -2\alpha_j \cdot \frac{1}{18}\fnorm{A_j} \\
            ={}& \fnorm{A_j}^2 + \frac{1}{324 \cdot 2^n}\fnorm{A_j}^2 - \frac{2}{324 \cdot 2^n}\fnorm{A_j}^2 \\
            ={}& \beta^2\fnorm{A_j}^2 \\
            \leq{}& \beta^{2(j+1)}\fnorm{A}^2, \tag{by \cref{eq:hypothesis_A}}
        \end{align*}
        which implies
        \begin{equation*}
            \fnorm{A - \sum_{k=0}^j\alpha_k V_k} \leq \beta^{j+1}\fnorm{A}
        \end{equation*}
        as desired. Moreover, using \cref{eq:alpha_j,eq:hypothesis_alpha}, we have 
        \begin{align*}
            \sum_{k=0}^j \alpha_k ={}& \sum_{k=0}^{j-1} \alpha_k + \alpha_{j}
            \leq{} \frac{1-\beta^j}{18\cdot 2^n(1-\beta)}\fnorm{A}+\frac{\beta^{j}}{18\cdot 2^n}\fnorm{A}
            ={} \frac{1-\beta^{j+1}}{18\cdot 2^n(1-\beta)}\fnorm{A},
        \end{align*}
        which completes the proof. \qedhere
\end{proof}

\subsection{Synthesizing through robust block-Hamiltonian simulation}
In this subsection, we first review the robust block-Hamiltonian simulation via QSVT~\cite{GSLW19}, then demonstrate how our approximation method from the previous subsection fits into this Hamiltonian simulation framework.

\begin{theorem}
    [{Robust block-Hamiltonian simulation, adapted from~\cite[extended version, Corollary~62]{GSLW19}}]\label{thm:hamiltonian_simulation}
    Let $t\in \RR$, $\varepsilon\in \rbra{0,1}$, $n \in \NN_+$ and let $U$ be an $(n+a)$-qubit unitary oracle such that
    \begin{equation*}
        \Abs*{\mathsf{H} - \alpha(\bra{0^a}\otimes I)U(\ket{0^a}\otimes I)} \leq \frac{\varepsilon}{\abs{2t}}
    \end{equation*}
    for an $n$-qubit Hamiltonian $\mathsf{H}$, $\alpha \in \RR_+$ and $a \in \NN$.
    Then, we can implement an $(n+a+2)$-qubit Hamiltonian simulation unitary $V$ such that
    \begin{equation*}
        \Abs*{e^{-i\mathsf{H}t} - (\bra{0^{a+2}}\otimes I)V(\ket{0^{a+2}}\otimes I)} \leq \varepsilon
    \end{equation*}
    with $O\rbra*{\alpha \abs{t}+\log\rbra{1/\varepsilon}}$ uses of $U$ or its inverse, $O(1)$ uses of controlled-$U$ or its inverse, using $O\rbra*{a\rbra*{\alpha\abs{t}+\log\rbra*{1/\varepsilon}}}$ two-qubit gates and using $O\rbra{1}$ ancillas.
\end{theorem}

Consequently, the T-count and ancilla-count analysis of $e^{-i\mathsf{H}t}$ reduces to examining the T-count of oracle $U$ in \cref{thm:hamiltonian_simulation} under the approximation by \cref{lem:arbitrary_operator_approximation}.

\begin{lemma}\label{lem:hamiltonian_t_count}
    Let $\varepsilon\in \rbra{0,1}$, $n \in \NN_+$ and $\mathsf{H}$ be an $n$-qubit Hamiltonian.
    Then, there is a unitary $U$ that can be implemented by Clifford+T circuits using
    \begin{equation*}
        O\rbra*{2^{n}\log\rbra{\max\cbra*{\Abs{\mathsf{H}}_F,1}/\varepsilon}}
    \end{equation*}
    T gates and ancillas such that $\Abs{\mathsf{H} - \alpha(\bra{0^a}\otimes I)U(\ket{0^a}\otimes I)} \leq \varepsilon$ with $\alpha = O\rbra{\Abs{\mathsf{H}}_F}$ and $a = O(n+\log\log\rbra{\max\cbra*{\Abs{\mathsf{H}}_F,1}/\varepsilon})$.
    The same holds for the controlled version of $U$.
\end{lemma}
\begin{proof}
    \allowdisplaybreaks
    By applying \cref{lem:arbitrary_operator_approximation} to the Hamiltonian $\mathsf{H}$ and choosing
    \begin{equation}\label{eq:m}
        m = \lceil 648\cdot 2^n\ln(\max\cbra*{\Abs{\mathsf{H}}_F,1}/\varepsilon_1) \rceil,
    \end{equation}
    there exist $n$-qubit Boolean phase oracles $B_1^{(k)}, B_2^{(k)}$, $n$-qubit diagonal unitaries $D^{(k)}$, and non-negative numbers $\alpha_k$ for $k=0, 1, \ldots, m-1$ such that
    \begin{equation}\label{eq:h_decomposition}
        \begin{aligned}
            \fnorm{\mathsf{H} - \sum_{k=0}^{m-1}\alpha_k B_1^{(k)}H^{\otimes n}D^{(k)}H^{\otimes n}B_2^{(k)} }
            \leq{}& \rbra*{\sqrt{1-\frac{1}{324\cdot 2^n}}}^{m}\Abs{\mathsf{H}}_F \\
            ={}& e^{m \ln\rbra*{1-\frac{1}{324\cdot 2^n}}/2}\Abs{\mathsf{H}}_F \\
            \leq{}& e^{-\frac{m}{648\cdot 2^n}} \Abs{\mathsf{H}}_F \\
            \leq{}& e^{-\ln(\max\cbra*{\Abs{\mathsf{H}}_F,1}/\varepsilon_1)}\Abs{\mathsf{H}}_F \\
            \leq{}& \varepsilon_1
        \end{aligned}
    \end{equation}
    and, setting $\beta = \sqrt{1-\frac{1}{324\cdot 2^n}}$,
    \begin{equation}
        \label{eq:h_alpha}
        \begin{aligned}
            \alpha = \sum_{k=0}^{m-1} \alpha_k
            &\leq \frac{1-\beta^m}{18\cdot 2^n\rbra*{1-\beta}}\Abs{\mathsf{H}}_F 
            \leq \frac{1}{18\cdot 2^n\rbra*{1-\beta}}\Abs{\mathsf{H}}_F 
            = \frac{1+\beta}{18\cdot 2^n\rbra*{1-\beta^2}}\Abs{\mathsf{H}}_F 
            \leq 36\Abs{\mathsf{H}}_F.
        \end{aligned}
    \end{equation}
    Let $a = \lceil\log_2(m)\rceil$, and suppose $G$ and $F$ satisfy
    \begin{align*}
        G\ket{0^a} ={}& \sum_{k=0}^{m-1} \sqrt{\frac{\alpha_k}{{\alpha}}}\ket{k}, \\
        F ={}& \sum_{k=0}^{m-1} \ketbra{k}{k}\otimes \rbra*{B_1^{(k)}H^{\otimes n}D^{(k)}H^{\otimes n}B_2^{(k)}} + \sum_{k=m}^{2^a-1}\ketbra{k}{k}\otimes I,
    \end{align*}
    then, by \cref{eq:h_decomposition}, we have
    \begin{equation}\label{eq:h_block_encoding}
        \begin{aligned}
            & \Abs*{\mathsf{H} - {\alpha} \rbra{\bra{0^a}\otimes I}\rbra{G^{\dagger}\otimes I}F\rbra{G\otimes I}\rbra{\ket{0^a}\otimes I}} \\
            ={}& \Abs*{\mathsf{H} - {\alpha} \rbra{\bra{0^a}G^{\dagger}\otimes I}F\rbra{G\ket{0^a}\otimes I}}\\
            ={}& \Abs*{\mathsf{H} - \sum_{k=0}^{m-1}\alpha_k B_1^{(k)}H^{\otimes n}D^{(k)}H^{\otimes n}B_2^{(k)}} \\
            \leq{}& \Abs*{\mathsf{H} - \sum_{k=0}^{m-1}\alpha_k B_1^{(k)}H^{\otimes n}D^{(k)}H^{\otimes n}B_2^{(k)}}_F \\
            \leq{}& \varepsilon_1.
        \end{aligned}
    \end{equation}
    We then figure out how to use Clifford+T circuits to approximate $G$ and $F$:
    \begin{itemize}
        \item By \cref{lem:optimal_preparation}, there is a state preparation oracle $\hat{G}$ implemented by a Clifford+T circuit using
        \begin{equation}\label{eq:G_approximate_T}
            O\rbra*{\sqrt{2^a\log\rbra*{1/\varepsilon_2}}+\log\rbra*{1/\varepsilon_2}}
        \end{equation}
        T gates and ancillas such that 
        \begin{equation}\label{eq:G_approximate}
            \Abs{G\ket{0^a} - \hat{G}\ket{0^a}} \leq \varepsilon_2.
        \end{equation}
        \item For $F$, notice that
            \begin{align*}
            F ={}& \sum_{k=0}^{m-1} \ketbra{k}{k}\otimes B_1^{(k)}H^{\otimes n}D^{(k)}H^{\otimes n}B_2^{(k)}  + \sum_{k=m}^{2^a-1}\ketbra{k}{k}\otimes I\\
            ={}& \rbra*{\bigoplus_{k=0}^{m-1} B_1^{(k)}H^{\otimes n}D^{(k)}H^{\otimes n}B_2^{(k)}} \oplus \rbra*{\bigoplus_{k=m}^{2^a-1} I_{2^n} H^{\otimes n}I_{2^n} H^{\otimes n}I_{2^n}}\\
            ={}& \rbra*{\bigoplus_{k=0}^{2^a-1} B_1^{(k)}}\rbra*{\bigoplus_{k=0}^{2^a-1} H^{\otimes n}} \rbra*{\bigoplus_{k=0}^{2^a-1} D^{(k)}} \rbra*{\bigoplus_{k=0}^{2^a-1} H^{\otimes n}}\rbra*{\bigoplus_{k=0}^{2^a-1} B_2^{(k)}}.
        \end{align*}
        By setting $B_1^{(k)} =D^{(k)} = B_2^{(k)} = I_{2^n}$ for $m \leq k\leq 2^a-1$, we obtain
        \begin{align*}
            F ={}& \rbra*{\bigoplus_{k=0}^{2^a-1} B_1^{(k)}} \rbra*{I_{2^a}\otimes H^{\otimes n}} \rbra*{\bigoplus_{k=0}^{2^a-1} D^{(k)}} \rbra*{I_{2^a}\otimes H^{\otimes n}} \rbra*{\bigoplus_{k=0}^{2^a-1} B_2^{(k)}}.
        \end{align*}
        Since ${\bigoplus_{k=0}^{2^a-1} B_1^{(k)}}$ and ${\bigoplus_{k=0}^{m-1} B_2^{(k)}}$ are $(n+a)$-qubit Boolean phase oracles, by \cref{lem:optimal_boolean_phase}, they can be implemented exactly by Clifford+T circuits using $O\rbra*{\sqrt{2^{n+a}}}$ T gates and ancillas.
        Similarly, ${\bigoplus_{k=0}^{2^a-1} D^{(k)}}$ is an $(n+a)$-qubit diagonal unitary, by \cref{lem:optimal_diagonal}, it can be implemented up to error $\varepsilon$ by a Clifford+T circuit using $O\rbra*{\sqrt{2^{n+a}\log\rbra*{1/\varepsilon}}+\log\rbra*{1/\varepsilon}}$ T gates and ancillas.
        Summing them up yields a unitary $\hat{F}$ implemented by a Clifford+T circuit using
        \begin{equation}\label{eq:F_approximate_T}
            O\rbra*{\sqrt{2^{n+a}\log\rbra*{1/\varepsilon_3}}+\log\rbra*{1/\varepsilon_3}}
        \end{equation}
        T gates and ancillas such that
        \begin{equation}\label{eq:F_approximate}
            \Abs{\hat{F} - F} \leq \varepsilon_3.
        \end{equation}
    \end{itemize}
    Therefore, setting $U = \rbra{\hat{G}^{\dagger}\otimes I}\hat{F}\rbra{\hat{G}\otimes I}$, we have 
    \begin{align*}
        & \Abs*{\mathsf{H} - {\alpha} \rbra{\bra{0^a}\otimes I}U\rbra{\ket{0^a}\otimes I}} \\
        ={}& \Abs*{\mathsf{H} - {\alpha} \rbra{\bra{0^a}\hat{G}^{\dagger}\otimes I}\hat{F}\rbra{\hat{G}\ket{0^a}\otimes I}} \\
        \leq{}& \Abs*{\mathsf{H} - {\alpha} \rbra{\bra{0^a}{G}^{\dagger}\otimes I}{F}\rbra{{G}\ket{0^a}\otimes I}} \\
        & \quad + \big\lVert {\alpha} \big(\rbra{\bra{0^a}{G}^{\dagger}\otimes I}{F}\rbra{{G}\ket{0^a}\otimes I}  - \rbra{\bra{0^a}\hat{G}^{\dagger}\otimes I}\hat{F}\rbra{\hat{G}\ket{0^a}\otimes I}\big)\big\rVert \\
        \leq{}& \varepsilon_1 + \alpha\big\lVert  \rbra{\bra{0^a}{G}^{\dagger}\otimes I}{F}\rbra{{G}\ket{0^a}\otimes I}  - \rbra{\bra{0^a}\hat{G}^{\dagger}\otimes I}\hat{F}\rbra{\hat{G}\ket{0^a}\otimes I}\big\rVert  \tag{by \cref{eq:h_decomposition}} \\
        \leq{}& \varepsilon_1 + \alpha\Abs*{\rbra{\bra{0^a}{G}^{\dagger}\otimes I} - \rbra{\bra{0^a}\hat{G}^{\dagger}\otimes I}} + \alpha\Abs*{F - \hat{F}} + \alpha \Abs*{\rbra{{G}\ket{0^a}\otimes I} - \rbra{\hat{G}\ket{0^a}\otimes I}}\\
        \leq{}& \varepsilon_1 +  {\alpha}\rbra*{2\varepsilon_2+\varepsilon_3} \tag{by \cref{eq:G_approximate,eq:F_approximate}} \\
        \leq{}& \varepsilon_1 + 36\Abs{\mathsf{H}}_F\rbra*{2\varepsilon_2+\varepsilon_3}. \tag{by \cref{eq:h_alpha}}
    \end{align*}
    Then, choosing
    \begin{equation}\label{eq:eps_123}
        \varepsilon_1 = \frac{\varepsilon}{2} \quad \text{and}\quad \varepsilon_2 = \varepsilon_3 = \frac{\varepsilon}{{108\Abs{\mathsf{H}}_F}},
    \end{equation}
    we obtain
    $\Abs*{\mathsf{H} - {\alpha} \rbra{\bra{0^a}\otimes I}U\rbra{\ket{0^a}\otimes I}} \leq \varepsilon$ and $a = O(n+\log\log\rbra{\max\cbra*{\Abs{\mathsf{H}}_F,1}/\varepsilon})$.
    The T-count and ancilla-count of $U$ is upper bounded by \cref{eq:G_approximate_T,eq:F_approximate_T} as
    \begin{align*}
        & O\rbra*{\sqrt{2^a\log\rbra*{1/\varepsilon_2}}+\log\rbra*{1/\varepsilon_2}} + O\rbra*{\sqrt{2^{n+a}\log\rbra*{1/\varepsilon_3}}+\log\rbra*{1/\varepsilon_3}} \\
        ={}& O\rbra*{\sqrt{2^{n+a}\log\rbra{\max\cbra*{\Abs{\mathsf{H}}_F,1}/\varepsilon}}+\log\rbra{\Abs{\mathsf{H}}_F/\varepsilon}} \tag{by \cref{eq:eps_123}} \\
        ={}& O\rbra*{2^{n}\log\rbra{\max\cbra*{\Abs{\mathsf{H}}_F,1}/\varepsilon}}. \tag{by \cref{eq:m} and $a = \lceil\log(m)\rceil$}
    \end{align*}
    For the controlled-$U$,
    we only need to implement the controlled-$F$ (without the need of controlled-$G$), which is
    \begin{align*}
        & \ketbra{0}{0}\otimes I_{2^a} \otimes I + \ketbra{1}{1} \otimes \Big(\sum_{k=0}^{m-1} \ketbra{k}{k}\otimes \rbra*{B_1^{(k)}H^{\otimes n}D^{(k)}H^{\otimes n}B_2^{(k)}} + \sum_{k=m}^{2^a-1}\ketbra{k}{k}\otimes I\Big) \\
        ={}& \sum_{k=0}^{m-1} \ketbra{1}{1}\otimes\ketbra{k}{k}\otimes \rbra*{B_1^{(k)}H^{\otimes n}D^{(k)}H^{\otimes n}B_2^{(k)}} + \rbra*{I_{2^{a+1}} - \sum_{k=0}^{m-1} \ketbra{1}{1}\otimes\ketbra{k}{k}}\otimes I
    \end{align*}
    and shares the same structure with $F$, but with one additional qubit.
    Then, the controlled-$F$ can be implemented up to error $\varepsilon_3$ by a Clifford+T circuit using the same T-count and ancilla-count as \cref{eq:F_approximate_T}.
    Consequently, the controlled-$U$ has the same T-count and ancilla-count as $U$.
\end{proof}

By combining the robust Hamiltonian simulation (\cref{thm:hamiltonian_simulation}) with the low T-count oracle construction (\cref{lem:hamiltonian_t_count}), we arrive at our main result.

\begin{theorem}[Low T-count unitary synthesis through robust block-Hamiltonian simulation]
    \label{thm:main}
    For any $n$-qubit Hamiltonian ${\mathsf{H}}$, its simulation unitary $e^{-i{\mathsf{H}}}$ can be implemented up to error $\varepsilon < 1$ in diamond distance by a Clifford+T circuit using
    \begin{equation*}
        O\rbra*{\rbra*{2^{n}+\log\log\rbra*{k/\varepsilon}}\rbra*{k+n+\log(1/\varepsilon)}\rbra*{n+\log\rbra*{k/\varepsilon}}}
    \end{equation*}
    T gates and $O\rbra*{2^n\rbra*{n+\log\rbra*{k/\varepsilon}}}$ ancillas, where $k = \max\cbra*{\Abs{\mathsf{H}}_F,1}$.
\end{theorem}
\begin{proof}
    By \cref{lem:hamiltonian_t_count}, there is a unitary $U$ that can be implemented by Clifford+T circuits using
    \begin{equation}\label{eq:t_count_u}
        O\rbra*{2^{n}\log\rbra{k/\varepsilon_1}}
    \end{equation}
    T gates and ancillas such that
    \begin{equation*}
        \Abs{\mathsf{H} - \alpha(\bra{0^a}\otimes I)U(\ket{0^a}\otimes I)} \leq \varepsilon_1/2
    \end{equation*}
    with
    \begin{equation}\label{eq:a_and_alpha}
        \alpha = O\rbra{k} \quad \text{and} \quad a = O(n+\log\log\rbra{k/\varepsilon_1}),
    \end{equation}
    and the same holds for the controlled version of $U$.
    Setting $t = 1$ in \cref{thm:hamiltonian_simulation}, we can implement a Hamiltonian simulation unitary $V$ such that
    \begin{equation*}
        \Abs{e^{-i\mathsf{H}} - (\bra{0^{a+2}}\otimes I)V(\ket{0^{a+2}}\otimes I)} \leq \varepsilon_1
    \end{equation*}
    with $O\rbra*{\alpha+\log\rbra{1/\varepsilon_1}}$ uses of $U$ or its inverse, $O(1)$ uses of controlled-$U$ or its inverse, using $O\rbra*{a\rbra*{\alpha+\log\rbra*{1/\varepsilon_1}}}$ two-qubit gates and using $O\rbra{1}$ ancillas.
    By \cref{lem:optimal_single_qubit}, we can implement each two-qubit gate to error $O\rbra*{\varepsilon/\rbra*{a\rbra*{\alpha+\log\rbra*{1/\varepsilon_1}}}}$ in operator norm by Clifford+T circuits using
    \begin{equation}\label{eq:t_count_two-qubit}
        O\rbra*{\log\rbra*{a\rbra*{\alpha+\log\rbra*{1/\varepsilon_1}}/\varepsilon_1}} = O\rbra*{\log\rbra*{a\alpha/\varepsilon_1}}
    \end{equation}
    T gates, which yields a unitary $V'$ such that
    \begin{equation}\label{eq:v'_block_encoding}
        \Abs{e^{-i\mathsf{H}} - (\bra{0^{a+2}}\otimes I)V'(\ket{0^{a+2}}\otimes I)} \leq 2\varepsilon_1.
    \end{equation}
    Its T-count is upper bounded by
    \begin{align*}
        & O(\alpha+\log(1/\varepsilon_1)) \cdot \underbrace{O\rbra*{2^{n}\log\rbra*{k/\varepsilon_1}}}_{\mathclap{\text{T-count of $U$ (\cref{eq:t_count_u})}}}  + O\rbra*{a\rbra*{\alpha+\log\rbra*{1/\varepsilon_1}}} \cdot \underbrace{O\rbra*{\log\rbra*{a\alpha/\varepsilon_1}}}_{\mathclap{\text{T-count of two-qubit gates (\cref{eq:t_count_two-qubit})}}} \\
        ={}& O\rbra*{\rbra*{\alpha+\log(1/\varepsilon_1)}\cdot \rbra*{2^{n}\log\rbra*{k/\varepsilon_1} + a\log\rbra*{a\alpha/\varepsilon_1}}} \\
        ={}& O\big(\rbra*{k+\log(1/\varepsilon_1)} \rbra*{2^{n}\log\rbra*{k/\varepsilon_1} + a\log\rbra*{ak/\varepsilon_1}}\big) \tag{by \cref{eq:a_and_alpha}} \\
        ={}& O\big(\rbra*{k+\log(1/\varepsilon_1)} \rbra*{2^{n}\log\rbra*{k/\varepsilon_1} + a\rbra*{\log\rbra*{a} + \log\rbra*{k/\varepsilon_1}}}\big) \\
        ={}& O\big(\rbra*{k+\log(1/\varepsilon_1)} (2^{n}\log\rbra*{k/\varepsilon_1} + a\rbra*{\log\rbra*{n+\log\log\rbra*{k/\varepsilon_1}} + \log\rbra*{k/\varepsilon_1}})\big) \tag{by \cref{eq:a_and_alpha}} \\
        ={}& O\big(\rbra*{k+\log(1/\varepsilon_1)} \rbra*{2^{n}\log\rbra*{k/\varepsilon_1} + a\rbra*{\log\rbra*{n} + \log\rbra*{k/\varepsilon_1}}}\big) \\
        ={}& O\big(\rbra*{k+\log(1/\varepsilon_1)}(2^{n}\log\rbra*{k/\varepsilon_1} + \rbra*{n+\log\log\rbra*{k/\varepsilon_1}}\rbra*{\log\rbra*{n} + \log\rbra*{k/\varepsilon_1}})\big) \tag{by \cref{eq:a_and_alpha}} \\
        ={}& O\big(\rbra*{k+\log(1/\varepsilon_1)} \rbra*{2^{n}\log\rbra*{k/\varepsilon_1} + \log\log\rbra*{k/\varepsilon_1}\log\rbra*{k/\varepsilon_1}}\big) \\
        ={}& O\big(\rbra*{2^{n}+\log\log\rbra*{k/\varepsilon_1}} \rbra*{k+\log(1/\varepsilon_1)}\log\rbra*{k/\varepsilon_1}\big),
    \end{align*}
    and its ancilla-count, upper bounded by \cref{eq:t_count_u}, is $O\rbra*{2^n\log\rbra*{k/\varepsilon_1}}$.

    Choosing $\varepsilon_1 = \varepsilon/2^{n+3}$ and applying \cref{lem:block_encoding_diamond} to \cref{eq:v'_block_encoding}, we have that $V'$ approximates $e^{-i\mathsf{H}}$ up to error $\varepsilon$ in diamond distance with T-count
    \begin{equation*}
        O\rbra*{\rbra*{2^{n}+\log\log\rbra*{k/\varepsilon}}\rbra*{k+n+\log(1/\varepsilon)}\rbra*{n+\log\rbra*{k/\varepsilon}}}
    \end{equation*}
    and ancilla-count $O\rbra*{2^n\rbra*{n+\log\rbra*{k/\varepsilon}}}$.
\end{proof}

The above Hamiltonian simulation result naturally extends to arbitrary unitaries through their Frobenius norm distance to the Clifford group.

\begin{corollary}
    Any $n$-qubit unitary $U$ can be implemented up to error $\varepsilon < 1$ in diamond distance by a Clifford+T circuit using
    \begin{equation*}
        O\rbra*{\rbra*{2^{n}+\log\log\rbra*{k/\varepsilon}}\rbra*{k+n+\log(1/\varepsilon)}\rbra*{n+\log\rbra*{k/\varepsilon}}}
    \end{equation*}
    T gates and $O\rbra*{2^n\rbra*{n+\log\rbra*{k/\varepsilon}}}$ ancillas, where $k = \max\cbra*{d_F^{\mathcal{C}}(U),1}$.
\end{corollary}
\begin{proof}
    By the definition of $d_F^{\mathcal{C}}(U)$, there are an $n$-qubit Clifford unitary $C$ and a real number $\theta$ such that
    \begin{equation}\label{eq:du}
        d_F^{\mathcal{C}}(U) = \Abs{U - e^{i\theta}C}_F = \Abs{e^{-i\theta}C^{\dagger}U - I}_F.
    \end{equation}
    Let $\mathsf{H} = -i\ln\rbra*{e^{-i\theta}C^{\dagger}U}$ with eigenvalues in $(-\pi,\pi]$ and consider its spectral decomposition $\mathsf{H} = \sum_j \alpha_j \ketbra{j}{j}$.
    Then
    \begin{align*}
        \Abs{\mathsf{H}}_F^2 = \sum_j \abs{\alpha_j}^2 &\leq \sum_j \abs{\pi\sin\rbra{\alpha_j/2}}^2 \tag{$\abs{x}\leq\pi\abs{\sin\rbra*{x/2}}$ for any $x\in \sbra*{-\pi,\pi}$} \\
        &= \frac{\pi^2}{4}\sum_j\abs{e^{-i\alpha_j}-1}^2 \\
        &= \frac{\pi^2}{4} \Abs{e^{-i\mathsf{H}}-I}_F^2
        =  \frac{\pi^2}{4} \Abs{e^{-i\theta}C^{\dagger}U-I}_F^2
        = \frac{\pi^2}{4} \rbra*{d_F^{\mathcal{C}}(U)}^2, \tag{by \cref{eq:du}}
    \end{align*}
    which means $\Abs{\mathsf{H}}_F \leq \frac{\pi}{2} d_F^{\mathcal{C}}(U)$.
    Since $U = e^{i\theta}C e^{-i\mathsf{H}}$, applying \cref{thm:main} to $\mathsf{H}$ yields the stated T-count and ancilla-count complexity for $e^{-i\mathsf{H}}$, and hence for $U$.
\end{proof}

Finally, in the important special case of unitaries with bounded Frobenius norm distance to the Clifford group, we achieve the following near-optimal scaling.

\begin{corollary}
    \label{cor:optimal_t_count}
    Any $n$-qubit unitary $U$ with $d_F^{\mathcal{C}}(U) \leq O(1)$ can be implemented up to error $\varepsilon < 1$ in diamond distance by a Clifford+T circuit using
    \begin{equation*}
        O\rbra*{\rbra*{2^{n}+\log\log\rbra*{1/\varepsilon}}\rbra*{n+\log(1/\varepsilon)}^2}
    \end{equation*}
    T gates and $O\rbra*{2^n\rbra*{n+\log\rbra*{1/\varepsilon}}}$ ancillas.
\end{corollary}

\section{Lower bound}

In this section,  we establish a ${\Omega}(2^n)$ T-count lower bound for implementing $n$-qubit unitaries $e^{-i\mathsf{H}}$ with $\Abs{\mathsf{H}}_F \leq 1$.
This represents a stronger result than the ${\Omega}(2^n)$ T-count lower bound for general $n$-qubit unitaries in~\cite{gosset2024quantumstatepreparationoptimal}.
Our goal is thus to prove the following theorem.

\begin{theorem}\label{thm:lower}
    For any integer $n \geq 1$ and any $0 < \varepsilon < 2^{-n}$, there is an $n$-qubit Hamiltonian $\mathsf{H}$ with $\Abs{\mathsf{H}}_F \leq O\rbra{1}$ such that any adaptive Clifford+T circuit that implements $e^{-i\mathsf{H}}$ to precision $\varepsilon$ in the diamond norm distance requires $\Omega\rbra{2^n\sqrt{\log\rbra{1/\varepsilon}} + \log\rbra{1/\varepsilon}}$ T gates.
\end{theorem}
\begin{proof}
    Following the same reasoning as \cite[Theorem~4.3 and Claim~4.11]{gosset2024quantumstatepreparationoptimal}, it suffices to prove \cref{prop:lower_postselection} below.
\end{proof}

An immediate corollary of \cref{thm:lower} thus provides a T-count lower bound for the synthesis of unitary operators $U$ when the distance $d_F^{\mathcal{C}}\rbra{U}$ of $U$ to the Clifford group is a constant. 

\begin{corollary}
    For any $n \geq 1$ and any $0 < \varepsilon < 2^{-n}$, there is an $n$-qubit unitary operator $U$ with $d_F^{\mathcal{C}}\rbra{U} \leq O\rbra{1}$ such that any adaptive Clifford+T circuit that implements $U$ to precision $\varepsilon$ in the diamond norm distance requires $\Omega\rbra{2^n\sqrt{\log\rbra{1/\varepsilon}} + \log\rbra{1/\varepsilon}}$ T gates.
\end{corollary}
\begin{proof}
    Let $\mathsf{H}$ be the $n$-qubit Hamiltonian specified in \cref{thm:lower}.
    Let $U = e^{-i\mathsf{H}}$, then $d_F^{\mathcal{C}}\rbra{U} \leq \Abs{\mathsf{H}}_F \leq O\rbra{1}$. 
    Therefore, any adaptive Clifford+T circuit that implements $U = e^{-i\mathsf{H}}$ to precision $\varepsilon$ in the diamond norm distance requires $\Omega\rbra{2^n\sqrt{\log\rbra{1/\varepsilon}} + \log\rbra{1/\varepsilon}}$ T gates.
\end{proof}

We now set out to toward the proof of \cref{thm:lower}, through auxiliary definitions and lemmas.

\begin{definition}
    [Choi state]
    For an $n$-qubit unitary $U$, its Choi state $\ket{\iota_U}$ is defined by
    \begin{equation*}
        \ket{\iota_U} = \frac{1}{\sqrt{2^n}}\sum_{j\in \cbra{0,1}^n}\ket{j}\otimes \rbra*{U\ket{j}}.
    \end{equation*}
\end{definition}

\begin{lemma}
    [{Adapted from~\cite[Fact~4.14]{gosset2024quantumstatepreparationoptimal}}]\label{lem:t2f_distance}
    Let $U$ and $V$ be two $n$-qubit unitaries. The trace distance between their Choi states is
    \begin{align*}
        & \frac{1}{2}\Abs*{\ketbra{\iota_U}{\iota_U} - \ketbra{\iota_V}{\iota_V}}_1 
        \geq{} \frac{1}{\sqrt{2^{n+1}}}\cdot \min_{\theta\in [0,2\pi)}\Abs{U-e^{i\theta}\cdot V}_F.
    \end{align*}
\end{lemma}

\begin{proposition}\label{prop:lower_postselection}
    For any integer $n \geq 1$ and sufficiently small $\varepsilon > 0$, there is an $n$-qubit Hamiltonian ${H}$ with $\Abs{{H}}_F \leq 1$ such that the following holds.
    Assume $\mathcal{C}$ is a Clifford circuit with $m$ Pauli postselections and $a$ ancillas, $\mathcal{C}\rbra*{\ket{0^{2n}}\ket{T}^{\otimes t}\ket{0^a}}=\ket{\phi}\ket{0^{t+a}}$, where $\ket{T} = \frac{1}{\sqrt{2}}\rbra*{\ket{0}+e^{i\pi/4}\ket{1}}$ is the magic state, and $\frac{1}{2}\Abs*{\ketbra{\phi}{\phi}-\ketbra{\iota_{e^{-iH}}}{\iota_{e^{-iH}}}}_1 \leq \varepsilon$. Then, $t = \Omega\rbra*{2^n\sqrt{\log\rbra*{1/\varepsilon}-n} + \log\rbra{1/\varepsilon}}$. 
\end{proposition}

We will use the following sphere packing result in arbitrary norms.
Here, the volume is the standard Euclidean one.

\begin{lemma}[Adapted from~\cite{schildkraut2024lower}]\label{lem:sphere_packing}
    For sufficiently large integer $d$, an arbitrary norm $\norm[\mathit{arb}]{\cdot}$ over $\RR^d$, and any compact $\Omega \subset \RR^d$, we can find a finite set $A \subset \Omega$ such that
    \begin{equation*}
        \abs{A} > \rbra*{1-o(1)}\frac{d\ln d}{2^{d+1}}\Vol(\Omega)
    \end{equation*}
    and the balls $B(x,r_d)$ and $B(y,r_d)$ are disjoint for distinct $x,y \in A$, where
    \[B(x,r) = \{z\in \RR^d \mid \norm[\mathit{arb}]{x-z} \leq r\}\]
    and $r_d$ is the radius such that $\Vol\rbra*{B(0,r_d)} = 1$. 
\end{lemma}
\begin{proof}
    See the proof of Theorem~1.1 in~\cite{schildkraut2024lower}.
\end{proof}

\begin{lemma}\label{lem:hamiltonian_count}
    For any integer $n \geq 1$ and sufficiently small $\varepsilon > 0$, there are $N = \rbra*{\sqrt{2^n}\varepsilon}^{-\Omega(4^n)}$ $n$-qubit Hermitian operators $\mathsf{H}_1, \ldots, \mathsf{H}_N$ with each $\Abs{\mathsf{H}_j}_F \leq 1$  such that the pairwise trace distance between Choi states $\ket{\iota_{e^{-i{\mathsf{H}}_1}}}, \ldots,\ket{\iota_{e^{-i{\mathsf{H}}_N}}}$ is at least $\varepsilon$.
\end{lemma}
\begin{proof}
    We start by choosing $n$-qubit Hermitian operators $\mathsf{H}_1,\ldots,\mathsf{H}_M$ with $\Abs{\mathsf{H}}_F \leq 1/2$ individually and pairwise distance
    $\Abs{e^{-i\mathsf{H}_j} - e^{-i\mathsf{H}_k}}_F \geq 10\sqrt{2^{n+1}}\varepsilon$.
    By \cite[Equation~(71)]{gosset2024quantumstatepreparationoptimal}, we have $\Abs{e^{-i\mathsf{H}_j} - e^{-i\mathsf{H}_k}}_F \geq \Abs{\mathsf{H}_j - \mathsf{H}_k}_F/3$, then it is sufficient to make the pairwise distance
    $\Abs{\mathsf{H}_j - \mathsf{H}_k}_F \geq 30\sqrt{2^{n+1}}\varepsilon$.
    Since $n$-qubit Hermitian matrices are determined by $4^n-1$ independent real parameters and can be embedded into $\RR^{4^n-1}$,
    we treat them as vectors in $\RR^{4^n-1}$ endowed with norm $\fnorm{\cdot}$.
    Then, we apply \cref{lem:sphere_packing} with $d = 4^n-1$ to get a set
    \begin{equation*}
        A \subset B\rbra*{0, \frac{r_{d}}{60\sqrt{2^{n+1}}\varepsilon}}
    \end{equation*}
    such that
    \begin{equation*}
        \abs{A} > \rbra*{1-o(1)}\frac{d\ln\rbra*{d}}{2^{d+1}}\Vol\rbra*{B\rbra*{0, \frac{r_d}{60\sqrt{2^{n+1}}\varepsilon}}} \\
    \end{equation*}
    and $\Abs{\mathsf{H}_j - \mathsf{H}_k} \geq 2r_d$ for distinct $\mathsf{H}_j, \mathsf{H}_k \in A$.
    Scaling with factor $60\sqrt{2^{n+1}}\varepsilon/(2r_d)$, we can get
    \begin{equation*}
        A' \subset B\rbra*{0, \frac{1}{2}} = \cbra*{\mathsf{H} \bigm| \Abs{\mathsf{H}}_F \leq \frac{1}{2}, \mathsf{H}\in \Herm\rbra*{2^n} }
    \end{equation*}
    such that
    \begin{align*}
        \abs{A'} >{}& \rbra*{1-o(1)}\frac{d\ln\rbra*{d}}{2^{d+1}}\Vol\rbra*{B\rbra*{0, \frac{r_d}{60\sqrt{2^{n+1}}\varepsilon}}} \\
        ={}& \rbra*{1-o(1)}\frac{d\ln\rbra*{d}}{2^{d+1}}\Vol\rbra*{B\rbra*{0, r_d}}\rbra*{\frac{1}{60\sqrt{2^{n+1}}\varepsilon}}^{d} \\
        ={}& \rbra*{1-o(1)}\frac{d\ln\rbra*{d}}{2^{d+1}}\rbra*{\frac{1}{60\sqrt{2^{n+1}}\varepsilon}}^{d} \tag{by definition of $r_d$ in \cref{lem:sphere_packing}} \\
        ={}& \Omega\rbra*{\rbra*{\frac{1}{60\sqrt{2^{n+1}}\varepsilon}}^d} \\
        \geq{}& \rbra*{\sqrt{2^n}\varepsilon}^{-\Omega(4^n)} \tag{by $d = 4^n-1$}
    \end{align*}
    and $\Abs{\mathsf{H}_j - \mathsf{H}_k} \geq D$ for distinct $\mathsf{H}_j, \mathsf{H}_k \in A'$.
    Thus, we can choose $M = \rbra*{\sqrt{2^n}\varepsilon}^{-\Omega(4^n)}$ and traceless $n$-qubit Hermitian operators $\mathsf{H}_1,\ldots,\mathsf{H}_M$ with $\Abs{\mathsf{H}}_F \leq 1/2$ individually and pairwise distance
    \begin{equation}\label{eq:unitary_f_distance}
        \Abs{e^{-i\mathsf{H}_j} - e^{-i\mathsf{H}_k}}_F \geq 10\sqrt{2^{n+1}}\varepsilon.
    \end{equation}
    By \cref{lem:t2f_distance}, for distinct $j,k \in \sbra{M}$, we have
    \begin{align*}
        \frac{1}{2}\Abs*{\ketbra{\iota_{e^{-i\mathsf{H}_j}}}{\iota_{e^{-i\mathsf{H}_j}}} - \ketbra{\iota_{e^{-i\mathsf{H}_k}}}{\iota_{e^{-i\mathsf{H}_k}}}}_1 
        \geq{}& \frac{1}{\sqrt{2^{n+1}}}\cdot \min_{\theta\in [0,2\pi)}\Abs{e^{-i\mathsf{H}_j}-e^{i\theta}\cdot e^{-i\mathsf{H}_k}}_F.
    \end{align*}
    Then, by applying the discretisation trick for rotation angle $\theta$ in the proof of~\cite[Fact~4.7]{gosset2024quantumstatepreparationoptimal} with \cref{eq:unitary_f_distance}, each $\ket{\iota_{e^{-i\mathsf{H}_j}}}$ can be $\varepsilon$-close to only $O\rbra*{1/\varepsilon}$ other $\ket{\iota_{e^{-i\mathsf{H}_k}}}$'s in trace distance.
    Hence we can find $N = \Omega\rbra*{\varepsilon M} = $ Choi states $\ket{\iota_{e^{-i{\mathsf{H}}_1'}}}, \ldots,\ket{\iota_{e^{-i{\mathsf{H}}_N'}}}$ from $\ket{\iota_{e^{-i{\mathsf{H}}_1}}}, \ldots,\ket{\iota_{e^{-i{\mathsf{H}}_M}}}$ such that their pairwise trace distance is at least $\varepsilon$, which completes the proof.
\end{proof}

\begin{lemma}
    [Adapted from~{\cite[Proof of Proposition~4.6]{gosset2024quantumstatepreparationoptimal}}]\label{lem:state_count}
    Let $\mathcal{C}$ be a Clifford circuit with $m$ Pauli postselections, $n$ input qubits, and $a$ ancillas.
    Then, there are at most $2^{O\rbra{n^2+t^2}}$ possible $n$-qubit states $\ket{\phi}$ such that $\ket{\phi}\ket{0^{t+a}} = \mathcal{C}\rbra*{\ket{0^n}\ket{T}^{\otimes t}\ket{0^a}}$.
\end{lemma}

We are now ready to prove \cref{prop:lower_postselection}.

\begin{proof}[Proof of \cref{prop:lower_postselection}] 
    The lower bound $\Omega\rbra{\log\rbra{1/\varepsilon}}$ follows from~\cite[Lemma~5.9]{beverland2020lower} for single qubit unitaries.

    By \cref{lem:hamiltonian_count}, there are at least $N = \rbra*{\sqrt{2^n}\varepsilon}^{-\Omega(4^n)}$ $n$-qubit Hamiltonians ${H}_1, \ldots, {H}_N$ with each $\Abs{{H}_j}_F \leq 1$ such that the pairwise trace distance between Choi states $\ket{\iota_{e^{-i{H}_1}}}, \ldots,\ket{\iota_{e^{-i{H}_N}}}$ is greater than $2\varepsilon$.
    
    On the other hand, by \cref{lem:state_count}, to prepare each $2n$-qubit Choi state $\ket{\iota_{e^{-i{H}_j}}}$ up to error $\varepsilon$ in trace distance, we need $2^{O\rbra{n^2+t^2}} \geq \rbra*{\sqrt{2^n}\varepsilon}^{-\Omega(4^n)}$, which gives $t = \Omega\rbra*{2^n\sqrt{\log\rbra*{1/\varepsilon}-n}}$.
\end{proof}

\section*{Acknowledgment}
Large language model tools were used to assist with exposition, literature search, and exploratory discussions of possible approaches to the research problem. 
The authors reviewed and edited the manuscript as needed and take full responsibility for its content.

\addcontentsline{toc}{section}{References}

\bibliographystyle{alphaurl}
\bibliography{main}

\end{document}